\documentclass[%
 aip,
cp,  
amsmath,amssymb,
 reprint,%
]{revtex4-2}
\usepackage{amsthm}
\usepackage{graphicx}
\usepackage{dcolumn}
\usepackage{bm}

\usepackage[utf8]{inputenc}
\usepackage[T1]{fontenc}
\usepackage{mathptmx}

\newtheorem{lemma}{Lemma}
\newtheorem{corollary}{Corollary}

\newtheorem{definition}{Definition}
\newtheorem{claim}{Claim}
\newcommand{\F}{\mathbb{F}}

\begin{document}

\title{On The Number of Different Entries in Involutory MDS Matrices over Finite Fields of Characteristic Two}

\author{Muhammad Afifurrahman} 
 \email[Corresponding author: ]{m\_afifurrahman@students.itb.ac.id}
\affiliation{
  Algebra Research Group, Faculty of Mathematics and Natural Sciences, Bandung Institute of Technology
  (Jalan Ganesa 10 Bandung, Indonesia).
}

\date{\today} 

\begin{abstract}
Two of many criteria of a good MDS matrix are being involutory and having few different elements. This paper investigates the number of different entries in an involutory MDS matrices of order 1, 2, 3, and 4 over finite fields of characteristic two. There are at least three and four different elements in an involutory MDS matrices with, respectively, order three and four, over finite fields of characteristic two.

\end{abstract}

\maketitle

\section{\label{sec:level1}Introduction}

MDS matrices, especially over finite fields of characteristic two, are widely used in cryptography for constructing block ciphers due to its diffusion properties. For an extensive survey, we consult Gupta et. al. \cite{gupta}.

Two of many criteria of a good MDS matrix are being involutory and having few different elements. Junod and Vaudenay \cite{junod} have found some lower bounds (with examples) on the numbers of different entries in MDS matrices over finite fields of characteristic two. In this paper, we extend this result to involutory MDS matrices.

This paper is organized as follows: second section introduces MDS matrices and some relevant results. Third section concerns the lower bounds (with examples) of the number of different entries in an involutory MDS matrices with order 1, 2, and 3 over finite fields of characteristic two. Fourth section concerns lower bounds (with examples) of the number of different entries in an involutory MDS matrices with order 4 over finite fields of characteristic two. The fifth section summarizes the results of this paper.
\section{Definition and Lemma}
Let $\F_{2^m}$ be the finite field with $2^m$ elements, and $m$ be a positive integer. We refer to $n\times n$ matrices as matrices of order $n$. A square matrix $A$ is involutory if $A^2=I$. Denote $A_{ij}$ as the entry of row $i$ and column $j$ in a matrix $A$.
\begin{definition}\cite{gupta}
		A matrix $A$ is \textit{MDS} (Maximum Distance Separable) if every square submatrices of $A$ are nonsingular.
\end{definition}



The following lemma, and its corollaries, play an important part in this paper.
\begin{lemma}\cite{gupta}
For any permutation matrices (with correct sizes) $P$ and $Q$, if $A$ is MDS, $PAQ$ and $A^T$ is also MDS.
\end{lemma}

Based on this, and owing to the fact that $P^{-1}$ is also a permutation matrix, the following corollaries are obvious.
\begin{corollary} 
\label{paap}
A matrix $A$ is involutory MDS if and only if $PAP^{-1}$ is involutory MDS, with $P$ being a permutation matrix.
\end{corollary}

\begin{corollary}\label{transpose}
A matrix $A$ is involutory and MDS if and only if $A^T$ is involutory and MDS.
\end{corollary}

\section{Matrices of Order One, Two, and Three}
Obviously, there are at least one different element in an involutory MDS matrices of order 1, and $\begin{pmatrix}
1
\end{pmatrix}^2=\begin{pmatrix}
1
\end{pmatrix}$ 

For order two, obviously $\begin{pmatrix}
a&a\\a&a
\end{pmatrix}$ is not MDS - hence, any MDS matrices need at least two different elements. Now, for any $a\in \F_{2^m}, a\notin \{0,1\}$,  $\begin{pmatrix}
a&a+1\\a+1&a
\end{pmatrix}^2=\begin{pmatrix}
1&0\\0&1
\end{pmatrix}$.  Hence, at least two different elements are needed in an involutory MDS matrices. Furthermore, any involutory MDS matrices over $\F_{2^m}$ with order two and exactly two different elements is of the form $\begin{pmatrix}
a&a+1\\a+1&a
\end{pmatrix}$ - hence, there are exactly $2^m-2$ matrices in this case.

The case of order three is more involved. First, any MDS matrices in this order has at least two different elements. Suppose there was an involutory MDS matrix which have exactly two different elements, denoted by $a$ and $b$. Combinatorial reasoning gives two possible families of MDS matrices, up to the transformation described in corollary \ref{paap}, - $\begin{pmatrix}
a&a&b\\a&b&a\\b&a&a
\end{pmatrix}$ and $\begin{pmatrix}
a&a&b\\b&a&a\\a&b&a
\end{pmatrix}$. Choosing appropriate $a$ and $b$ (such that none of the square submatrices are singular) results in an MDS matrix with exactly two different entries. But, were one of these matrices be involutory, by checking each entries, it can be inferred that one of $a$ and $b$ is 0 - a contradiction. Hence, at least three different entries is needed in an involutory MDS matrices over $\F_{2^m}$. Examples of an involutory MDS matrices (over $\F_{2^3}$) with exactly three different entries can be seen at G{\"u}zel et. al \cite{guzel}.

\section{Matrices of Order Four}

Junod and Vaudenay \cite{junod} have already proved that any MDS matrices over $\F_{2^m}$ with order 4 has at least three different entries, with lower bound attained (as example) from matrices used in AES \cite{junod}. However, this matrix is not involutory.

Now suppose $A$ is an involutory MDS matrices with exactly three different entries $a,b,c\in \F_{2^m}$. Three successive claims (and sub-claims) are proven to classify the structure of $A$.

\begin{claim}\label{first}
$a,b,$ and $c$ appear at most twice in any row, or column, of $A$.
\end{claim}

\begin{proof}
Without loss of generality, and considering corollary \ref{paap} and \ref{transpose}, it is sufficient to prove they appear at most twice in $A$'s first row.

Suppose there exists a matrix $A$ that satisfy the initial condition such that its first row contains $a$ more than twice. By pigeonhole principle (applied over second row), $a$ must not appear four times. Hence, $a$ appears exactly thrice. Let another entry in $A$'s first row be $b$. By applying corollary \ref{paap}, there are only two form of $A$'s first row that needs to be considered: $\begin{pmatrix}
b&a&a&a
\end{pmatrix}$ and  $\begin{pmatrix}
a&a&a&b
\end{pmatrix}$.

\paragraph{First Case: $\begin{pmatrix}
b&a&a&a
\end{pmatrix}$.} We look at submatrix of $A$ made by removing the first column and row of $A$. Each entry in each row of this submatrix is pairwise different, because $A$ is MDS. Hence, for $i=2,3,4$, $\{A_{i2},A_{i3},A_{i4}\}=\{a,b,c\}$ and $A_{i2}+A_{i3}+A_{i4}=a+b+c$. Hence, $\displaystyle \sum_{i=2}^{4}\sum_{j=2}^{4} A_{ij} = 3(a+b+c)=a+b+c$.

Meanwhile, considering  $(A^2)_{1i}$ for $i=2,3,4$, we get $ba+aA_{2i}+aA_{3i}+aA_{4i}=0 \implies A_{2i}+A_{3i}+A_{4i}=b$. Adding all equation, $\displaystyle \sum_{i=2}^{4}\sum_{j=2}^{4} A_{ij} = b$. Hence, $a+b+c=b\implies a=c$, a contradiction.

\paragraph{Second Case:$\begin{pmatrix}
a&a&a&b
\end{pmatrix}$.} By the same argument as the first case, for $i=2,3,4$,  $\{A_{i1},A_{i2},A_{i3}\}=\{a,b,c\}$ and $A_{i1}+A_{i2}+A_{i3}=a+b+c$. Now, for $j=2,3,4$, if $A_{j4}=b$, $A$ has $\begin{pmatrix}
a&b\\a&b
\end{pmatrix}$ as submatrix - contradicting $A$ being MDS. Hence, $A_{j4}\in \{a,c\}$. By looking at $A_{24}$ and $A_{34}$, there are two cases to be considered:
\begin{itemize}
\item Let $A_{24}=A_{34}=x$. By considering $(A^2)_{14}$, we get $ab+aA_{24}+aA_{34}+bA_{44}=0 \implies ab=bA_{44} \implies A_{44}=a$.

Then, considering $(A^2)_{24}+(A^2)_{34}$, we get $(A_{21}+A_{31})b+x(A_{22}+A_{23}+A_{32}+A_{33})=0$. Because $A_{i1}+A_{i2}+A_{i3}=a+b+c$ for $i=2,3,4$, this equation is equivalent to $(A_{21}+A_{31})b+x[(a+b+c+A_{21})+(a+b+c+A_{31})]=0$ and $(A_{21}+A_{31})(b+x)=0$.

Since $x\neq b$, $A_{21}=A_{31}$. But, the submatrix $\begin{pmatrix}
A_{21}&x\\A_{31}&x
\end{pmatrix}$ is singular, a contradiction.
\item $A_{24}\neq A_{34}$. Either $(A_{24},A_{34})=(a,c)$ or $(c,a)$. In both cases, by considering $(A^2)_{14}$, we get $ab+aa+ac+bA_{44}=0$. $A_{44}=a$ results in $a=0$ or $a=c$, and $A_{44}=c$ results in $a=b$ or $a=c$, a contradiction.
\end{itemize}

Because all cases leads to contradictions, the first statement must be true. Hence, the claim is proven.
\end{proof}

\begin{claim}\label{second}
Each row and column of $A$ must contain $a,b,$ and $c$.
\end{claim}

\begin{proof}
By the same argument as the last claim, it is sufficient to prove the first row of $A$ contains all of them. Suppose it is not. Without loss of generality (and by the last claim), let $a$ and $b$ appear twice in $A$'s first row, with $A_{11}=a$. By corollary \ref{paap}, let the first row of $A$ be $\begin{pmatrix}
a&b&b&a
\end{pmatrix}$. By considering the fourth row, we get $A_{42}\neq A_{43}$. Then, we exclude some possible values that they can take.

\begin{claim}
Neither $A_{42}$ nor $A_{43}$ are $b$.
\end{claim}

\begin{proof}
Suppose the claim is false, and without loss of generality, let $A_{42}=b$. Considering $(A^2)_{12}$, we get $ab+bA_{22}+bA_{32}+ab=0 \implies A_{22}=A_{32}=x$. From claim \ref{first}, $x\neq b$.

Now, consider the fourth row. Were $A_{41}$ or $A_{44}$ be $a$, $A$ would have singular submatrix - a contradiction. With the same reasoning, $A_{41}\neq A_{44}$. Hence $(A_{41},A_{44})=(b,c)$ or $(c,b)$ and $A_{41}+A_{44}=b+c$.

On the other hand, considering $(A^2)_{42}$, we get $A_{41}b+bx+A_{43}x+bA_{44}=0 \implies b(b+c)+(b+A_{43})x=0$. If $A_{43}=c$, $x=b$. If $A_{43}=b, b=0$ or $b=c$. Both leads to contradictions - hence, $A_{43}=a$.

Considering $(A^2)_{13}$, we get $ab+b(A_{23}+A_{33})+a^2=0$. If $(A_{23},A_{33})$ is a permutation of $(a,b)$, the last equation is equivalent to $a^2+b^2=0\implies a=b$ - a contradiction. If it is a permutation of $(a,c)$, the equation implies $a^2+bc=0$. Now, consider submatrix of $A$ constructed by taking first and fourth row, and taking third column and either of first and fourth column (depending of which of $A_{41}$ and $A_{44}$ is $c$). This submatrix is a permutation of $\begin{pmatrix}
b&a\\a&c
\end{pmatrix}$. Then, $A$ would have a singular submatrix - a contradiction.  Hence, $(A_{23},A_{33})$ is a permutation of $(b,c)$, and $ab+bc+b^2+a^2=0 \label{important}$.

Back to $(A^2)_{42}$, where we get $b(b+c)+(b+a)x=0$. If $x=c$, this equation is equivalent to $b^2+ac=0$. Combining with the last equation, we get $(a+b)(a+c)=0$ and $a=b$ or $a=c$ - a contradiction. Hence $x=a$.

Considering $(A^2)_{43}$, $bA_{41}+A_{23}b+A+A_{33}a+aA_{44}=0$. By the previous observations, $A_{23}=A_{44}$ or $A_{23}=A_{41}$. Were the first be true, $A_{33}=A_{41}$ and $(a+b)(A_{23}+A_{33})=0$. This implied $A_{23}=A_{33}$ or $a=b$, a contradiction. Hence, $A_{23}=A_{41}=p$ and $A_{33}=A_{44}=q$, for $(p,q)=(b,c)$ or $(c,b)$.

Considering $(A^2)_{32}$, $a(a+A_{33})+b(A_{31}+A_{34})=0$. Were $A_{31}=A_{34}$, $a=A_{33}$; a contradiction. Hence, they are different. If they are permutation of $(a,b)$, the equation is equivalent to $a^2+b^2+ab+aA_{33}=0$. But, from the previous paragraph, $ab+bc+b^2+a^2=0$; hence $bc=aA_{33}$ and $A_{33}=a$ - a contradiction. If they are permutation of $(a,c)$, the equation is equivalent to $a^2+ab+bc+aA_{33}=0\implies aA_{33}=b^2 \implies A_{33}=c$. But, $a^2+ab+bc+ac=0=(a+b)(a+c)\implies a=b$ or $a=c$, a contradiction. Hence, $(A_{31},A_{33})$ must be a permutation of $(b,c)$. This implies $a^2+b^2+bc+aA_{33}=0$ and $A_{33}=b=q$. Hence, $p=c$ and $A=\begin{pmatrix}
a&b&b&a\\A_{21}&a&c&A_{24}\\A_{31}&a&b&A_{34}\\c&b&a&b
\end{pmatrix}$.

By analyzing submatrices with order 2 of $A$, $(A_{31},A_{34})=(b,c)$. But, from $(A^2)_{33}$, $1=b^2+ac+b^2+ac=0$; a contradiction. Hence the initial assumption is false, and this claim is proved.
\end{proof}

From the claim, $(A_{42},A_{43})$ must be a permutation of $(a,c)$. Without loss of generality, let $(A_{42},A_{43})=(c,a)$.

From $(A^2)_{12}$, $ab+b(A_{22}+A_{32})+ac=0$. It can be seen that $A_{22}\neq A_{32}$. Furthermore, were $(A_{22},A_{32})$ be a permutation of $(b,c)$, the equation is equivalent to $(b+a)(b+c)=0$ and $a=b$ or $b=c$, a contradiction. Were they be a permutation of $(a,c)$, the equation implies $c(a+b)=0$ and $a=b$, also a contradiction. Hence they are permutation of $(a,b)$ and $b^2=ac$.

Now, from $(A^2)_{13}$, $ab+b(A_{23}+A_{33})+a^2=0$. It can be seen that $A_{23}\neq A_{33}$.  Furthermore, were $(A_{23},A_{33})$ be a permutation of $(a,b)$, the equation is equivalent to $a^2+b^2=0\implies a=b$, a contradiction. Were it be a permutation of $(b,c)$, the equation is equivalent to $ab+b^2+bc+a^2=0\implies (a+b)(a+c)=0\implies a=b$ or $a=c$, a contradiction. Hence, they are a permutation of $(a,c)$. And, $a^2=bc$.

Now consider $\begin{pmatrix}
A_{22}&A_{23}\\A_{32}&A_{33}
\end{pmatrix}$, a submatrix of $A$ Notice that from all possible matrices that can be obtained by the restrictions imposed above, its determinant is zero. Hence, $A$ has a singular submatrix - a contradiction. Then, the initial assumption (that the first row of $A$ contains  $a$ and $b$ only) is false. Hence, the claim is proven.
\end{proof}

\begin{claim}\label{last}
For all $i=1,2,3,4$, $\{A_{i1},A_{i2},A_{i3},A_{i4}\}-\{A_{ii}\}=\{a,b,c\}$ and $\{A_{1i},A_{2i},A_{3i},A_{4i}\}-\{A_{ii}\}=\{a,b,c\}$.
\end{claim}

\begin{proof}
It is sufficient to prove the first equality in case $i=1$. Suppose this claim is false. Considering the previous claim and corollary \ref{paap}, without loss of generality let the first row of $A$ be $\begin{pmatrix}
a&b&b&c
\end{pmatrix}$.

First we assert that $A_{22}\neq A_{32}$ and $A_{23}\neq A_{33}$.
Suppose this assertion is false. Without loss of generality (by considering corollary \ref{paap}), let $A_{22}=A_{32}=x$.

Considering $(A^2)_{12}$, we get $ab+cA_{42}=0$. It can be seen that $A_{42}=c$. By claim \ref{second}, $x$ is neither $b$ nor $c$; hence it is $a$.

Now consider $A$'s third column. It can be seen that $A_{23}\neq A_{33}$; furthermore, none of them is $a$. Hence, they are permutation of $(b,c)$ and from the claim \ref{second}, $A_{43}=a$.

However, by looking at $(A^2)_{13}$, $ab+b^2+bc+ca=0\implies (b+a)(b+c)=0\implies b=a$ or $b=c$, a contradiction. Hence the initial assertion is true.

For $i=2,3$, consider $(A^2)_{1i}$. We get the equation $ab+bA_{2i}+bA_{3i}+cA_{4i}=0$. From claim \ref{second}, one of $A_{2i},A_{3i},A_{4i}$ is $c$. If the first two element is not $c$, then it can be estabilished that $ab+b^2+ba+c^2=0\implies b=c$, a contradiction. Hence, one of $A+{2i}$ and $A_{3i}$ is $c$. Assume the other element is $b$. From claim \ref{second}, $A_{4i}$ must be $a$. However, from the equation above, $ab+b^2+bc+ca=0=(b+a)(b+c) \implies b=a$ or $b=c$. So, for $i=2,3$, $A_{2i}+A_{3i}=a+c$.

However, $(ab+bA_{22}+bA_{32}+cA_{42})+(ab+bA_{23}+bA_{33}+cA_{43}=0 \implies c(A_{42}+A_{43})=0\implies A_{42}=A_{43}$. By considering the submatrix $\begin{pmatrix}
b&b\\A_{42}&A_{43}
\end{pmatrix}$, this contradicts $A$ being MDS. Hence, the initial assumption is wrong, and the claim is proven. 
\end{proof}

Now, without loss of generality, let the first row of $A$ be $\begin{pmatrix}
a&a&b&c
\end{pmatrix}$. Then, we observe the second column. From claim \ref{last}, $(A_{32},A_{42})=(b,c)$ or $(c,b)$. 

Assume $(A_{32},A_{42})=(b,c)$. By reasoning based on claim \ref{last}, $A_{43}=a$, $A_{23}=c$, $A_{41}=b$, $A_{21}=a$, and $A_{31}=c$. Now, by considering $(A^2)_{11}$, we get $1=a^2+a^2+bc+cb=0$, a contradiction. So, $(A_{32},A_{42})=(c,b)$. By considering $(A^2)_{12}$, we get $a^2+aA_{22}+bc+cb=0\implies A_{22}=a$.

Now we observe the third row. By claim \ref{last}, because $A_{32}=c$, either $(A_{31},A_{34})=(a,b)$ or $(b,a)$.

If $(A_{31},A_{34})=(a,b)$, by reasoning based on claim \ref{last}, $A_{24}=a$, $A_{23}=c$, $A_{21}=b$, $A_{41}=c$, and $A_{43}=a$. Now, by considering $(A^2)_{13}$, $ab+ac+bA_{33}+ac=0 \implies a=A_{33}$. However, by considering $(A^2)_{23}$, we get $b^2+ac+ca+a^2=0 \implies b=a$, a contradiction.

Hence, $(A_{31},A_{34})=(b,a)$. By considering $(A^2)_{32}$, we get $ab+ac+cA_{33}+ab=0 \implies a=A_{33}$. By considering $A$'s submatrix from first \& second row and first \& second column, we get $A_{21}\neq a$. By claim \ref{last}, $A_{21}=c$ and $A_{41}=a$. But, by considering $(A^2)_{31}$, we get $0=ab+c^2+ba+a^2 \implies c=a$, a contradiction. Hence, there are no involutory MDS matrices over $\F_{2^m}$ that has exactly, or less than, three different elements.

We conclude that for any natural $m$, any involutory MDS matrix $A$ with elements in $\F_{2^m}$ has at least four different elements. An example of involutory MDS matrices (over $\F_{2^8}$) that has exactly four different entries is used in \textit{Anubis} block cipher \cite{anubis}.

\section{Conclusion}

Any involutory MDS matrices over $\F_{2^m}$ of order three and four need (respectively) three and four different elements. This result extends the result from Junod and Vaudenay \cite{junod}, which proves an MDS matrices (not needed to be involutory) over $\F_{2^m}$ of order three and four need at least two and three different elements.

\begin{acknowledgments}
The author would like to thank Aleams Barra and  Intan Muchtadi-Alamsyah for providing valuable suggestions. This research is supported by Hibah Riset Dasar DIKTI 2019.
\end{acknowledgments}

\bibliography{ref}
\end{document}